\documentclass[submission,copyright,creativecommons]{eptcs}
\providecommand{\event}{SOS 2007} % Name of the event you are submitting to
\usepackage{breakurl}             % Not needed if you use pdflatex only.
\usepackage{underscore}           % Only needed if you use pdflatex.
\usepackage{amsmath}
\usepackage{amssymb}
\usepackage{proof}
\usepackage{amsthm}

\newtheorem{theorem}{Theorem}
\newtheorem{lemma}[theorem]{Lemma}
\newtheorem{proposition}[theorem]{Proposition}
\theoremstyle{definition}
\newtheorem{definition}[theorem]{Definition}
\newtheorem{example}[theorem]{Example}
\newcommand{\Nat}{{\mathbb N}}

\usepackage[dvipsnames]{xcolor}

\newcommand{\Set}{\mathsf{Set}}
\newcommand{\ttt}{\mathsf{tt}}
\newcommand{\Bool}{2} %\mathbb{B}}
\newcommand{\truebool}{\mathsf{true}}
\newcommand{\falsebool}{\mathsf{false}}
\newcommand{\leftbool}{\mathsf{left}}
\newcommand{\rightbool}{\mathsf{right}}

\newcommand{\leaf}[1]{\mathsf{leaf}\;#1}
\newcommand{\node}[2]{\mathsf{node}\;#1\;#2}
\newcommand{\Tree}[2]{\mathsf{Tree}\;#1\;#2}
\newcommand{\TreeAgdaInd}[3]{\mathsf{Tree}\;#1\;#2\;#3}
\newcommand{\TreeAgdaCoind}[3]{\mathsf{Tree'}\;#1\;#2\;#3}
\newcommand{\Size}{\mathsf{Size}}
\newcommand{\SizeLt}[1]{\mathsf{Size}{<}\;#1}
\newcommand{\infPath}{\mathsf{diverge}}
\newcommand{\force}{\mathsf{force}}
\newcommand{\skipo}{\mathsf{sk}}
\newcommand{\mapTree}[2]{\mathsf{mapTree}\;#1\;#2}
\newcommand{\liftTree}[3]{{#1}^{\alpha}\;#2\;#3}
\newcommand{\coliftTree}[3]{{#1}^{\beta}\;#2\;#3}
\newcommand{\leafS}{\mathsf{leaf}}
\newcommand{\nodeS}{\mathsf{node}}
\newcommand{\TreeS}{\mathsf{Tree}}
\newcommand{\liftTreeS}[1]{{#1}^{\alpha}}
\newcommand{\coliftTreeS}[1]{{#1}^{\beta}}

\newcommand{\Test}[1]{\mathsf{Test}\;#1}
\newcommand{\atom}[1]{\mathsf{atom}\;#1}
\newcommand{\true}{\mathsf{True}}
\newcommand{\false}{\mathsf{False}}
\newcommand{\liftTest}[2]{\overline{\mathsf{Test}}\;#1\;#2}
\newcommand{\dualTest}[1]{\mathsf{dualTest}\;#1}
\newcommand{\mapTest}[2]{\mathsf{mapTest}\;#1\;#2}
\newcommand{\liftTestS}{\overline{\mathsf{Test}}}
\newcommand{\dualTestS}{\mathsf{dualTest}}
\newcommand{\mapTestS}{\mathsf{mapTest}}
\newcommand{\TestS}{\mathsf{Test}}

\newcommand{\coinductive}{\mathsf{coinductive}}
\newcommand{\data}{\mathsf{data}}
\newcommand{\record}{\mathsf{record}}
\newcommand{\field}{\mathsf{field}}
\newcommand{\where}{\mathsf{where}}

\newcommand{\leafalpha}[2]{\mathsf{leaf}_\alpha\;#1\;#2}
\newcommand{\nodealpha}[1]{\mathsf{node}_\alpha\;#1}
\newcommand{\pileaf}[1]{\pi_{\mathsf{leaf}}\;#1}
\newcommand{\pinode}[2]{\pi_{\mathsf{node}}\;#1\;#2}
\newcommand{\leafbeta}[1]{\mathsf{leaf}_\beta\;#1}
\newcommand{\exepbeta}[1]{\mathsf{exep}_\beta\;#1}
\newcommand{\nodebeta}[1]{\mathsf{node}_\beta\;#1}
\newcommand{\pid}{\pi_{\mathsf{d}}}

\newcommand{\nodealphaS}{\mathsf{node}_\alpha}
\newcommand{\pileafS}{\pi_{\mathsf{leaf}}}
\newcommand{\pinodeS}{\pi_{\mathsf{node}}}
\newcommand{\leafbetaS}{\mathsf{leaf}_\beta}
\newcommand{\exepbetaS}{\mathsf{exep}_\beta}
\newcommand{\nodebetaS}{\mathsf{node}_\beta}

\newcommand{\Or}{\mathsf{or}}
\newcommand{\lookup}{\mathsf{lookup}}
\newcommand{\update}{\mathsf{update}}
\newcommand{\inp}{\mathsf{input}}

\newcommand{\may}{\mathsf{may}}
\newcommand{\must}{\mathsf{must}}

\newcommand{\N}{\mathsf{N}}
\newcommand{\tto}{\Rightarrow}
\newcommand{\U}{\mathsf{U}}
\newcommand{\Ty}{\mathsf{Ty}}
\newcommand{\Val}{\mathsf{Val}}
\newcommand{\Cpt}{\mathsf{Cpt}}
\newcommand{\val}{\mathsf{val}}
\newcommand{\cpt}{\mathsf{cpt}}
\newcommand{\Tm}{\mathsf{Tm}}
\newcommand{\Formo}[2]{\mathsf{Fma}\;#1\;#2}
\newcommand{\eq}{\mathsf{eq}}
\newcommand{\neqF}{\mathsf{neq}}
\newcommand{\lef}{\mathsf{fst}}
\newcommand{\rig}{\mathsf{snd}}
\newcommand{\thunk}{\mathsf{thunk}}
\newcommand{\obsalpha}{\mathsf{obs}_\alpha}
\newcommand{\obsbeta}{\mathsf{obs}_\beta}
\newcommand{\test}{\mathsf{test}}
\newcommand{\negForm}[1]{\mathsf{negFma}\;#1}
\newcommand{\negFormS}{\mathsf{negFma}}

\title{Inductive and Coinductive Predicate Liftings for Effectful Programs}
\author{Niccol\`{o} Veltri \thanks{The authors were supported by the ESF funded Estonian IT Academy research measure (project 2014-2020.4.05.19-0001) and the Estonian Research Council grant PSG659.}
	\quad \qquad Niels F.W. Voorneveld \footnotemark[1]
\institute{Tallinn University of Technology\\ Tallinn, Estonia}
\email{niccolo@cs.ioc.ee \quad \qquad niels@cs.ioc.ee}
}
\def\titlerunning{Inductive and Coinductive Predicate Liftings for Effectful Programs}
\def\authorrunning{N. Veltri \& N.F.W. Voorneveld}
\begin{document}
\maketitle

\begin{abstract}
	We formulate a framework for describing behaviour of effectful higher-order recursive programs. Examples of effects are implemented using effect operations, and include: execution cost, nondeterminism, global store and interaction with a user. The denotational semantics of a program is given by a coinductive tree in a monad, which combines potential return values of the program in terms of effect operations. 
	
	Using a simple test logic, we construct two sorts of predicate liftings, which lift predicates on a result type to predicates on computations that produce results of that type, each capturing a facet of program behaviour. Firstly, we study inductive predicate liftings which can be used to describe effectful notions of total correctness. Secondly, we study coinductive predicate liftings, which describe effectful notions of partial correctness. The two constructions are dual in the sense that one can be used to disprove the other. 
	
	The predicate liftings are used as a basis for an endogenous logic of behavioural properties for higher-order programs. The program logic has a derivable notion of negation, arising from the duality of the two sorts of predicate liftings,
	%which due to the duality of the two sorts of predicate liftings, has a constructive notion of negation. The program logic
	and it generates a program equivalence which subsumes a notion of bisimilarity. Appropriate definitions of inductive and coinductive predicate liftings are given for a multitude of effect examples.
	
	The whole development has been fully formalized in the Agda
	proof assistant.
	
\end{abstract}
%\begin{keyword}
%	Induction, coinduction, effects, predicate lifting, program equivalence, bisimilarity
%\end{keyword}
%\end{frontmatter}

\section{Introduction}
\label{sec:intro}

Programs may exhibit different behaviour depending on various circumstances. The environment can induce an effect upon program evaluation in many ways, e.g. nondeterministic decision making, access to some global store, or evaluation timeouts.
Such effects can be represented as \emph{algebraic effect operations} \cite{effect}, which trigger an effect and capture all possible continuations. 

General recursive computations with effect triggering operations can be seen as generating evaluation traces in the form of \emph{coinductive trees}. Such trees have as leaves the potential results a computation can return, and as nodes effect operations which branch into each of the possible continuations. Due to the coinductive nature of these trees, they can be of infinite height, denoting potentially diverging computations.

In practice, when a program is evaluated, each of the effect operations is handled by choices made by the environment. An external observer, e.g. a user of a program, cannot directly inspect the generated tree. However, this spectator may make certain \emph{observations} capturing effectful behaviour \cite{observation}. Such observations include: possible termination, evaluation within a time limit, and how a program alters a global state.

We capture such observations with a collection of predicate liftings \cite{modal-journal,MatacheS19}. These consist of a set of tokens denoting each possible observation, and for each token a device that lifts a predicate on return values to a predicate on coinductive trees over these values. Using a simple test logic, these predicate liftings are described ``locally'' in terms of effect operations. The local definitions generate an inductive predicate lifting, which captures a notion of effectful termination, and is used to specify total correctness with respect to the appropriate observation. The predicate liftings are similar to the ones derived from ordered monads by Hasuo~\cite{HasuoGeneric}.

Dually, we also generate coinductive predicate liftings. These allow computations to satisfy observations either immediately, or by postponing the burden of proof, potentially indefinitely. As such, they capture a notion of effectful divergence, and can be used to specify partial correctness with respect to effect observations. 

A main contribution of this paper is formulating the coinductive predicate liftings in such a way that they can be used as a constructive notion of negation for the inductive predicates. As a result, these coinductive predicate liftings can be used to disprove inductive properties, and vice versa.
This complementing pair hence helps us recover in our constructive setting some of the expressive power available when reasoning classically as in \cite{modal-journal}.

From the inductive and coinductive predicate liftings, we generate a logic for higher-order programs, extending the logic from \cite{modal-journal} with coinductive predicate liftings (modalities). This logic captures behavioural properties (observations) for program denotations, and is \emph{endogenous} in the sense of Pnueli \cite{Pnueli77}.
It is possible to build nested statements, mixing inductive and coinductive properties.
The logic gives rise to a notion of program equivalence on program denotations, which subsumes a notion of applicative bisimilarity in the sense of \cite{Abramsky90,Relational}.

We use predicate liftings and the resulting logic to study programs exhibiting various different effects, including: evaluation cost, nondeterminism, global store, and user input. The adaptive nature of the generic formulations allows for many more examples to be implemented. As such, we believe it gives a foundation for the verification of effectful programs in many situations.

In particular, we study the preservation of the logical program equivalence over certain program transformations. One such transformation prevalent for higher-order programs is that of \emph{sequencing}, as used e.g. by \textsf{let}-binding, program composition and the uncurrying operation.
We show that, for our effect examples, a property for a sequenced program can be \emph{decomposed} into (or pulled back to) a property for unsequenced programs. This adapts the notion of \emph{decomposability} featured in \cite{modal-journal}, and gives us a variety of proof techniques for the verification of higher order programs.

The development in this paper has been fully formalized in the Agda proof assistant. The code is freely accessible at: \url{https://github.com/niccoloveltri/ind-coind-pred-lifts}. It uses Agda 2.6.1 with standard library 1.6-dev. One consequence of this formalization is that all the behavioural properties described by our predicate liftings and logic are fully \emph{constructive}.
Programming languages that can be deeply embedded in Agda can be studied using this logic too.
For instance, we have an implementation of an effectful version of fine-grained call-by-value PCF \cite{PCF,fine-grained} with its own variation on the logic similar to the one used in \cite{modal-journal}.

\textbf{Basic Type-Theoretic Definitions and Notation.}
Our work takes place in Martin-L{\"o}f type theory extended with inductive and coinductive types, and practically in the Agda proof assistant.
We use Agda notation for
dependent function types: $(x :A) \to B\;x$. We write $0$ for the
empty type and $1$ for the unit type with unique inhabitant $\ttt : 1$.
We use $\Bool$ for the type of Booleans with two inhabitants
$\truebool,\falsebool: \Bool$. We write $\{ x \}$ as a synonym for $1$
when we want to use a specific name $x$ instead of $\ttt$. Similarly, we write
$\{x,y\}$ for the synonym of $\Bool$ where $\truebool$ and
$\falsebool$ are replaced by $x$ and $y$.
We use $A^*$ for the type of finite lists of elements of type $A$.
Propositional equality is ${\equiv}$ and judgemental equality is ${=}$ (as in Agda).
Types are stratified in a cumulative hierarchy of universes
$\Set_k$. The first universe is simply denoted $\Set$ and when we
write statements like ``$X$ is a type'', we mean $X : \Set_k$ for some
universe level $k$. For readability reasons, in the paper all mentions
of universe levels have been removed. We often employ the proposition-as-types perspective, and write statements like ``$X$ implies $Y$'', which formally expresses the existence of a function from type $X$ to type $Y$, or ``$X$ holds'', meaning that there exists an element of type $X$.

\section{Programs as Trees}
\label{sec:programs}

Assume given a type $K$ and a type family $I : K \to \Set$. The pair
$(K,I)$ is frequently called a \emph{container}~\cite{AAG:concsp}, and
it is used as a \emph{signature}~\cite{monad} to specify branching of trees in type theory. 
\begin{definition}
	The type of
	\emph{coinductive trees} is generated by two constructors:
	\begin{equation}
	\hfill
	\label{eq:trees}
	\small
	\infer{\leaf x : \Tree I A}
	{x : A}
	\qquad
	\infer={\node k {ts} : \Tree I A}
	{k : K & ts : I\;k \to \Tree I A}
	\hfill \qquad
	\end{equation}
\end{definition}
A coinductive tree $t : \Tree I A$ can either be of the form $\leaf
x$, for $x$ a value of type $A$, or $\node{k}{ts}$, where $ts$
represents the immediate subtrees of $t$. The branching is specified
by the type $I\;k$. 
Following Leroy and Grall's notation \cite{LeroyG09}, we use 
the double line in the inference rule of
$\nodeS$ to indicate that this is a coinductive constructor, and
as such it can be iterated an infinite amount of times. For example,
assuming $\skipo : K$ with $I \;\skipo = 1$, the corecursive definition
of the infinite tree $\infPath : \Tree I A$ with exactly one branch at each level
goes as follows: $\infPath = \node {\skipo} {(\lambda x.\;\infPath)}$.

In Agda, coinductive types are encoded as \emph{coinductive records}
which can optionally be parametrized by a \emph{size} to ease the
productivity checking of corecursive
definitions~\cite{APTS:coppis,Dan:upttus}. For example, the type of
trees given above is implemented in Agda as the following pair of
mutually defined types:
\begin{equation}\label{eq:agdatrees}\def\arraystretch{0.8}{\small
	\!\!
	\begin{array}{l}
	\data \;\TreeAgdaInd {(I : K \to \Set)} {(A : \Set)} {(i : \Size)} : \Set \;\where\\
	\quad \leafS : (x : A) \to \TreeAgdaInd I A i \\
	\quad \nodeS : (k : K) \, (ts : I \,k \to \TreeAgdaCoind I A i) \to \TreeAgdaInd I A i
	\end{array} \
	\begin{array}{l}
	\record \; \TreeAgdaCoind  {(I : K \!\to\! \Set)} {(A : \Set)} {(i : \Size)} : \Set \;\where \\
	\quad  \coinductive \\
	\quad \field \\
	\quad \quad \force : \{j : \SizeLt i\} \to \TreeAgdaInd I A j
	\end{array}}
\end{equation}
The type $\TreeAgdaInd I A i$ is inductive, while $\TreeAgdaCoind I A
i$ is a coinductive record type. The inductive constructors
$\leafS$ and $\nodeS$ are similar to the inference rules in
(\ref{eq:trees}), but now the subtrees $ts$ in $\node{k}{ts}$ have
return type $\TreeAgdaCoind{I}{A}{i}$. Elements of the latter type are
like thunked computations that can be resumed by feeding a token, in
the form of a size $j$ smaller than $i$, to the destructor
$\force$. The resumed subtree inhabits the type
$\TreeAgdaInd{I}{A}{j}$. The decrease in size can be explained by
viewing sizes as abstract ordinals representing the number of
unfoldings that a coinductive definitions can undergo. Sizes come with
a top element $\infty : \Size$, and a fully formed coinductive tree
$t$ has type $\TreeAgdaInd I A \infty$. This means that a 
user has an infinite number of tokens that she can spend for accessing
arbitrarily deep subtrees of $t$. The presence of sizes in the types
of coinductive records is crucial for ensuring productivity of
corecursively defined functions, which would otherwise need to pass
Agda's strict syntactic productivity check. 

In this paper, we opt to
work with an informal treatment of coinductive types, specified by
constructors as in (\ref{eq:trees}). Accordingly, corecursive
functions are given as usual in terms of recursively defined functions. In
particular, all mentions to sizes in types have been dropped, but the
interested reader can recover them in our Agda formalization.

The $\TreeS\; I$ datatype is a functor, and we call $\mapTree$ its
action on functions. It is also a monad, with unit $\eta = \leafS$
and multiplication $\mu : \Tree I {(\Tree I A)} \to \Tree I A$
corecursively defined as: 
$$\mu\;(\leaf x) = x, \quad \mu\;(\node k {ts})
= \node k {(\lambda x.\; \mu\;(ts\;x))}.$$

In our programs-as-trees perspective, elements in a leaf position
represent return values of computations, while each node in a tree
denotes the presence of an effect operation. $K$ is the type of
admissible operations and, for each $k : K$, $I \;k$ is a type
corresponding to the arity of operation $k$. The existence of a
\emph{skip} operation $\skipo : K$ with $I\;\skipo = 1$ enables the
encoding of possibly non-terminating computations, as done in \cite{Capretta05}, like the diverging program
$\infPath$ introduced previously in this section. If $K$ is equivalent to
the unit type and $\skipo$ is its unique inhabitant, then $\Tree I A$
corresponds to the type of deterministic possibly non-terminating
programs returning values of type $A$ after some number of steps.

\section{Observations}
\label{sec:observations}

As discussed in the last paragraph of the previous section, a user of
a program represented by a tree in $\Tree I A$, with only one
admissible skip operation $\skipo : K$, is able to observe the
terminating and non-terminating behaviour of the program and its
computation time. A user of a program represented by a nondeterministic binary tree may
observe that some of the possible return values in the tree satisfy a
certain property, or that none of them do. In the case of computer
programs exhibiting a variety of different computational effects, the
user can observe other relevant behaviour and therefore make appropriate
queries to the system.

The collection of admissible queries on effectful programs in $\Tree I
A$, capturing observations of effectful behaviour, is given by a
particular class of tree predicates, generated by a
predicate tree lifting. 
We introduce an inductive grammar of
logical statements, whose elements we call \emph{tests}, to define these predicate liftings.
\begin{center}
	$ 
	\begin{array}{c}
	\infer{\atom a : \Test A}{a : A}
	\qquad
	\infer{\true : \Test A}{}
	\qquad
	\infer{\false : \Test A}{}
	\end{array}
	$
\end{center}
\begin{center}
	$
	\begin{array}{c}
	\infer{t \wedge u : \Test A}{t,u : \Test A}
	\qquad
	\infer{t \vee u : \Test A}{t,u : \Test A}
	\qquad
	\infer{\bigwedge t : \Test A}{t : \Nat \to \Test A}
	\qquad
	\infer{\bigvee t : \Test A}{t : \Nat \to \Test A}
	\end{array}
	$ 
\end{center}

Elements of $\Test{A}$ are interpreted as types via the 
lifting $\liftTestS : (A \to \Set) \to \Test A \to
\Set$. Atoms, i.e. tests of the form $\atom a$, are modelled as
$P\;a$ where $P : A \to \Set$ is the input predicate. The other
constructors of $\Test{A}$ are modelled via Curry-Howard
correspondence. E.g. conjunction of tests $t$ and $u$ is interpreted
as the Cartesian product of the interpretations of $t$ and $u$, and
similarly for all the other logical operations. The type former
$\TestS$ is a functor, we call $\mapTestS$ its action
on functions.

Observations on coinductive trees are formulated using a threefold
specification:
\begin{enumerate}
	\item A type $O$ of tokens, each denoting a particular observation
	or test on an effectful program.
	\item A decidable subset $\pileafS : O \to \Bool$ specifying which
	observations are satisfied by the production of a successful
	result. This is called the \emph{leaf function} for $O$.
	\item A function $\pinodeS : (k : K) \to O \to  \Test {(I \;k \times
		O)}$ formulating when a tree with a root node satisfies an
	observation. This is called the \emph{node function} for $O$. Note
	that the constructors of the grammar $\TestS$ allow for both
	parallel and sequenced testing.
\end{enumerate}

To each observation $o : O$ and tree $t : \Tree I \Set$, which we
think of as obtained from a tree in $\Tree I A$ after an application of a
predicate $P : A \to \Set$ to its leaves, we wish to associate a type
$\alpha\;o\;t : \Set$ corresponding to the \emph{handling} of $t$ \cite{Handlers}. The result should be a particular logical combination of the
truth values from $t$ dictated by the input observation $o$.
This is done in the following definition:
\begin{definition}
	Given $\pileaf : O \to 2$ and $\pinode : (l : K) \to O \to \Test{I \ k \times O}$, we define the $O$-indexed algebra $\alpha : O \to \Tree I \Set \to \Set$ inductively as follows:
	
	\begin{center}
		$
		\infer{\leafalpha b p : \alpha\;o \;(\leaf P)}
		{b : \pileaf o \equiv \truebool & \quad p : P}
		\qquad \qquad
		\infer{\nodealpha \textit{ps} : \alpha\;o \;(\node k \textit{Ps})}
		{\textit{ps} : \liftTest {(\lambda (x,o').\;\alpha\;o'\;(\textit{Ps}\;x))} {(\pinode k o)} }
		$ 
	\end{center}
\end{definition}
The satisfiability of the predicate $\alpha\;o$ by a tree $t : \Tree I \Set$
depends on the shape of $t$.  If $t = \leaf P$, for some $P : \Set$, then
it satisfies $\alpha\;o$ in case $P$ holds and the token $o$ is
correct according to the leaf function $\pileafS$. If $t = \node k
{Ps}$, for some $k : K$ and $Ps : I\;k \to \Tree I \Set$, then the
test $\pinode{k}{o}$ specifies a combination of continuations $x$ and
new observations $o'$. If subtrees $Ps\;x$ satisfy $\alpha\;o'$,
according to the logical formula corresponding to $\pinode{k}{o}$, then also $t$ satisfies $\alpha\;o$. The use of the inductively-defined grammar of logical statements $\TestS$ and its interpretation in types $\liftTestS$ in the premise of $\nodealphaS$ guarantees that $\alpha$ is a well-defined inductive type family and passes Agda strict positivity check.

\begin{definition}\label{def:ind-alg}
	The \emph{inductive $O$-indexed predicate lifting} of a tree given by $\alpha$ is defined as
	
	\begin{center}
		$
		\begin{array}{l l}
		\liftTree{(-)} : O \to (A \to \Set) \to \Tree I A \to \Set, & \qquad
		\liftTree o P t = \alpha \;o\;(\mapTree P t) .
		\end{array}
		$
	\end{center}
\end{definition}

\noindent
Here we see the predicate $P$ as some observable property on values of type $A$. For each observation $o$, $\liftTree o$ lifts this property to an extensionally observable property on computations of type $A$.

For each effect, an 
$O$-indexed predicate lifting is specified according to two design principles:
\begin{itemize} 
	\item All predicate liftings in the family capture a notion of observation which is extensionally observable by a user of the program.
	\item Together, the predicate liftings are powerful enough to express differences between programs that might be detectable after certain canonical program transformations.
\end{itemize}
Later in this paper, we will formulate a notion of behavioural equivalence specified by this family of observations.
The second design principle expresses the desire that this behavioural equivalence is preserved by certain program transformations. 
One such transformation of particular interest to us
is that of program sequencing.

\begin{example}\label{ex:pure1}
	We start with purely deterministic programs, represented by trees in $\Tree I A$ where $K = \{ \skipo \}$ and $I \; \skipo = 1$; that is, trees have only one node label $\skipo$ which has only one continuation. The type $\Tree I A$ is therefore equivalent to Capretta's delay monad applied to $A$ \cite{Capretta05}.
	The label $\skipo$ corresponds to a deterministic program step, which may or may not be observable by the user of the program.
	A family of observations is chosen depending on whether the skip is observable.
	
	If skips are unobservable, we can only observe \emph{termination}.
	We take $O = \{ {\downarrow} \}$ and define $\pileaf {{\downarrow}} = \truebool$, meaning that we observe ${\downarrow}$ when the computation terminates, and $\pinode  {\skipo} {{\downarrow}} = \atom{(\ttt , {\downarrow})}$, meaning that we keep observing ${\downarrow}$ on the continuation of a not-yet-terminated computation. 
	In this case, $\liftTree {{\downarrow}} P t$ is provable when $t$ eventually produces a leaf $a : A$ such that $P \,a$ holds.
	
	Alternatively, we may consider $\skipo$ to be observable, for instance as a measure of evaluation time.
	We use as observations $O=\Nat$ expressing time limits on termination.
	We take $\pileaf n = \truebool$, $\pinode \skipo 0 = \falsebool$ and $\pinode \skipo {(n+1)} = \atom{(\ttt, n)}$, making $\liftTreeS{n}$ observe termination within at most $n$ skips.
\end{example}

\begin{example}\label{ex:nond1}
	Consider an unpredictable scheduler making binary decisions for a program.
	Such computations are denoted by \emph{binary decision trees}, i.e. coinductive trees in $\Tree{I}{A}$ over signature $K = \{ \Or \}$, $I \; \Or = \{ \leftbool,\rightbool\}$, whose nodes display choices given by the $\Or$ operation.
	Due to the unpredictability of the decisions made, we interpret these choices as being resolved \emph{nondeterministically}.
	Nondeterminism in functional languages has been thoroughly studied by Ong \cite{Ong} and Lassen \cite{Lassen} and many others.
	
	Users of nondeterministic programs cannot observe decision paths. They may at most observe what \emph{may} be possible, and what \emph{must} happen.
	As such, we formulate two observations $O = \{\may, \must\}$, with accompanying functions $\pileaf x = \truebool$, $\pinode \Or \may = \atom{(\leftbool, \may)} \vee \atom{(\rightbool, \may)}$ and $\pinode \Or \must = \atom{(\leftbool, \must)} \wedge \atom{(\rightbool, \must)}$.
	In this case, $\liftTree \may P t$ holds if there is some sequence of resolutions of choices for which $t$ produces a result satisfying $P$.
	On the other hand, $\liftTree \must P t$ holds if $t$ is guaranteed to produce a result satisfying $P$, no matter the decisions made.
\end{example}

\begin{example}\label{ex:glob1}
	Suppose programs can consult some global store containing a natural number, using a $\lookup : K$ and an update operation $\update : \Nat \to K$.
	The lookup operation is countably branching, with $I \; \lookup = \Nat$, and continues according to the current state.
	For any $n : \Nat$, $\update \ n$ stores $n$ in the global store, and continues in a unique way, defining $I \; (\update \ n) = 1$.
	Computations in $\Tree I A$ express communication patterns between a program and a global store, potentially yielding a result of type $A$.
	
	We suppose that both the state before the execution of a program and the state after the execution of a program are observable.
	As observations, we use pairs $O = \Nat \times \Nat$ of initial and final states.
	With $\liftTree{(n,m)}{P}{t}$ we aim to express that, when program $t$ is evaluated with starting state $n$, it produces a result satisfying $P$ with the final state equal to $m$.
	To this end, we define: 
	\begin{itemize}
		\item $\pileaf{(n,n)} = \truebool$, but $\pileaf{(n,m)} = \falsebool$ if $n \neq m$, since termination does not change the state,
		\item $\pinode{(\update \ k)}{(n,m)} = \atom{(\ttt , (k,m))}$, \ \ and $\pinode{\lookup}{(n,m)} = \atom{(n, (n,m))}$.
	\end{itemize}
\end{example}

\begin{example}\label{ex:inpu1}
	In this last example, we consider computations which repeatedly ask input bits from a user.
	We have one operation $\inp : K$ of arity $I \; \inp = \{ \leftbool,\rightbool\}$.
	The user chooses inputs for the program, and we specify two sorts of observations. 
	For any list of bits $l \in 2^*$ of length $n$, the user can check:
	\begin{itemize}
		\item whether the computation terminates after the sequence of bits $l$ is entered;
		\item whether after entering the sequence of bits $l$, the computation requests another input.
	\end{itemize}
	Note that the user can only input the sequence one by one, and only if the computation keeps requesting more inputs until the bitlist is exhausted.
	Observations $O = 2 \times 2^*$ consist of a bit denoting which of the two tests are performed and the list of bits used in the test. We implement the observations as follows:
	
	$\pileaf{(b, l)} = (b \equiv \leftbool) \wedge (l \equiv \langle \rangle)$,
	\qquad \qquad
	$\pinode{\inp}{(b, l)} = \begin{cases} b \equiv \rightbool & \text{if} \ l \equiv \langle \rangle \\
	\atom{(\mathsf{head}\;l, (b, \mathsf{tail}\;l))} & \text{otherwise.} \end{cases}$
\end{example}

For any observation $o$ in the examples above, the predicate $\liftTreeS{o}\;P$ expresses some effectful variation on termination checking, such as termination within some time limit, decision invariant termination, and termination with a certain final state.
The particular relevant notion of termination is dependent on our interpretation of the effect.
We say that $\liftTree{o}{P}{t}$ captures a notion of \emph{effectful total correctness} according to observation $o$.

The inductive predicates generated by $\alpha$ are however unable to `detect' the dual to termination: divergence.
From the user's perspective, divergence is mainly the absence of termination, the indefinite continuation of a program.
It is not possible to extensionally confirm whether a program diverges.
For some programs however, like in the given example $\infPath$, we can prove intensionally that a program will go on forever.
Following this example's lead, we will express \emph{provable effectful divergence} next,
using coinductive predicate liftings.

\section{Co-observations}\label{sec:co}

To prove correctness of a program, we either show that it terminates and produces a correct result, or we show it requires the resolution of some (effect) operation, in which case we may postpone the burden of proof to after this operation has been resolved.
In total correctness, we must verify that the program will eventually terminate.
In partial correctness however, we can get away with postponing the burden of proof indefinitely.
In that case, if the program diverges, it is still considered correct, since it will never return an incorrect result.

We implement this notion of partial correctness via \emph{coinductive predicate liftings}, a dual to the inductive predicate liftings from last section.
This additionally provides a technique for disproving inductive properties.
For example, to disprove that a program produces an even number, we can either show that the program produces an odd number, or that the program diverges.
In other words, if it is partially incorrect that a program produces an even number, then we can conclude that it is not totally correct that it produces an even number.
Using coinduction, we set up this notion of partial correctness in a constructive way, with the additional motivation to use it as a constructive notion of `negation'.
However, due to decidability issues, this will never truly be a perfect negation, only a generically large constructive complement.

\begin{definition}\label{def:coin-alg}
	Given $\pileaf : O \to 2$ and $\pinode : (l : K) \to O \to \Test{I \ k \times O}$, we define the $O$-indexed algebra $\beta : O \to \Tree{I}{\Set} \to \Set$ \emph{coinductively}, using the following judgments:
	
	\begin{center}
		$
		\small
		\infer{\leafbeta p : \beta\;o \;(\leaf P)}
		{p : P}
		\qquad\qquad
		\infer{\exepbeta b : \beta\;o \;(\leaf P)}
		{b : \pileaf o \equiv \falsebool}
		\qquad\qquad
		\infer={\nodebeta {ps} : \beta\;o \;(\node k {Ps})}
		{ps : \liftTest {(\lambda (x,o').\;\beta\;o'\;(Ps\;x))} {(\pinode k o)}}
		$
	\end{center}
\end{definition}
When a program returning values of type $\Set$ terminates immediately, there are two possible ways of
satisfying $\beta \;o$: either the program returns an inhabited type (as in the constructor $\leafbetaS$)
or $o$ is incorrect according to $\pileafS$ (as in the constructor $\exepbetaS$).
Conceptually, $\exepbetaS$ states that we do not need to verify correctness of the result if we are not currently verifying an observation which we consider `correct for termination'.

For instance, in Example \ref{ex:glob1} of global store, $\beta \; (0,1)$ expresses the following partial correctness property: if with starting state $0$ the program terminates with final state $1$, then it produces an inhabited type.
This statement could be satisfied with an exception ``$\exepbetaS$'', which entails that the program terminates with a final state we are not currently checking for. In that case, the property is satisfied since the condition we are testing for at termination is not met. Hence $\beta \; (0,1) \; (\leaf{P})$ is inhabited.

The double line in the $\nodebetaS$ constructor reflects the coinductive nature of the predicate lifting.
In the Agda implementation, $\beta$ is formulated using mutual induction-coinduction, similarly to the encoding of coinductive trees in (\ref{eq:agdatrees}). In particular, sizes appear as extra arguments to ensure productivity of corecursive definitions.
From $\nodebeta{ps}$ we can extract a proof in terms of proofs on the continuation $Ps$. Since $\nodebetaS$ is a coinductive constructor, it is possible for a proof of $\beta\;o\;t$ to contain
an infinite amount of such extractions.
Hence, proofs of $\beta$ can be self-referential, and can refer to infinitely many nodes in the tree.

\begin{definition}
	The \emph{coinductive $O$-indexed predicate lifting} of a tree given by $\beta$ is defined as:
	
	\begin{center}
		$
		\begin{array}{l l}
		\coliftTree{(-)} : O \to (A \to \Set) \to \Tree I A \to \Set, & \qquad
		\coliftTree o P t = \beta \;o\;(\mapTree P t) .
		\end{array}
		$
	\end{center}
\end{definition}

The coinductive $O$-indexed predicate lifting specify notions of partial correctness. In an effect-free language, this amounts to saying that, if the program terminates, it produces a correct result.
For effects, it captures more general notions of partial correctness, checking correctness of a result only at termination and only under certain circumstances. These circumstances are given by observations $o : O$ for which $\pileaf{o} \equiv \truebool$. 

Partial correctness is useful both in cases in which termination can be checked independently, or when one is more interested in the safety of results at the expense of possible divergence.
Considering non-terminating programs as at least not dangerous, gives more freedom when trying to design safe programs, and the coinductive predicate liftings give tools for the verification of such programs in many different effectful situations.

\subsection{Coinduction as Negation}

As mentioned before, another main use for coinduction is as a notion of negation for inductive properties. As such, it also has uses when studying and reasoning about inductive predicates.
But coinduction cannot always directly be used this way. We need to modify its formulation slightly.

In this paper, we introduce effectful versions of termination and divergence, which result in effectful versions of total and partial correctness complementing each other.
Considering that effect observations are formulated using tests, in order to properly function as a notion of constructive complement, we need to formulate the complement or \emph{dual} of a test, called $\dualTest : \Test A \to \Test A$.

\begin{center}
	$
	\arraycolsep=1.4pt
	\begin{array}{lcl@{\quad\quad}lcl}
	\dualTest{(\atom{a})} &=& \atom{a} &
	\dualTest{(t \vee u)} &=& \dualTest{t} \wedge \dualTest{u}\\
	\dualTest{\true} &=& \false &
	\dualTest{(\bigwedge t)} &=& \bigvee (\lambda n. \dualTest{(t \; n)})\\
	\dualTest{\false} &=& \true &
	\dualTest{(\bigvee t)} &=& \bigwedge (\lambda n. \dualTest{(t \; n)}) \\  
	\dualTest{(t \wedge u)} &=& \dualTest{t} \vee \dualTest{u}
	\end{array}
	$
\end{center}

Following well-established results in logic, this function can be used to give a constructive complement.
This is formulated by showing that it lifts \emph{disjoint} predicates on $A$ to disjoint predicates on $\Test A$.

\begin{definition}
	Two predicates $P, Q : A \to \Set$ on $A$ are \emph{disjoint} if there is a proof of $(a : A) \to P \; a \to Q \; a \to 0$.
	A predicate lifting over $F: \Set \to \Set$ is a function $f : (A \to \Set) \to (F \, A \to \Set)$.
	Two predicate liftings $f, g : (A \to \Set) \to (F \, A \to \Set)$ are \emph{distinct} if for any two disjoint predicates $P$ and $Q$ on $A$, the pair of lifted predicates $f \; P$ and $g \; Q$ are disjoint.
\end{definition}

\begin{lemma}\label{lem:distest}
	The predicate liftings $\liftTestS$ and $(\lambda P,t.\,\liftTest{P}{(\dualTest{t})})$ on $\TestS$ are distinct.
\end{lemma}

We use this dualization of tests to motivate a particular specification for our formulation of $\coliftTreeS{(-)}$. If for $\liftTree{(-)}$ we use $\pileafS$ and $\pinodeS$, then we will use $\pileafS$ and $\lambda k , o.\; \dualTest{(\pinode{k}{o})}$ as specification for the formulation of $\coliftTreeS{(-)}$. In this case, the premise of $\nodebetaS$ in Definition \ref{def:coin-alg} unfolds to the type $\liftTest {(\lambda (x,o').\;\beta\;o'\;(Ps\;x))} {(\dualTest{(\pinode{k}{o})})}$.

\begin{definition}
	Given $\pileaf : O \to 2$ and $\pinode : (l : K) \to O \to \Test{I \ k \times O}$, and suppose $\alpha$ is specified using $(O,\pileaf,\pinode)$ following Definition \ref{def:ind-alg}, and $\beta$ is specified using $(O,\pileaf,\lambda k , o.\; \dualTest{(\pinode{k}{o})})$ following Definition \ref{def:coin-alg}, then we call  $(\liftTreeS{(-)}, \coliftTreeS{(-)})$ a \emph{complementing pair}.
\end{definition}

\noindent
Note that there are multiple complementing pairs for each set $O$, since the specifications in terms of $\pileaf$ and $\pinode$ may vary.
We do this in order to establish the following result.
\begin{proposition}\label{prop:dist-pred}
	Suppose $(\liftTreeS{(-)}, \coliftTreeS{(-)})$ is a complementing pair of $O$-indexed predicate liftings, then for any $o : O$, the predicate liftings $\liftTreeS{o}$ and $\coliftTreeS{o}$ on $\TreeS \;I$ are distinct.
\end{proposition}
\begin{proof}
	Let $P$ and $Q$ be two disjoint predicates on $A$, and let $t$ be a coinductive tree.
	We need to show that it is absurd to assume both $\liftTree{o}{P}{t}$ and $\coliftTree{o}{Q}{t}$.
	The proof proceeds by pattern matching on $t$.
	
	If $t = \leaf{a}$, then $\liftTree{o}{P}{t}$ must be proved by $\leafalpha{b}{p}$ for some $b : \pileaf{o} \equiv \truebool$ and $p : P \; a$.
	Similarly, $\coliftTree{o}{P}{t}$ must be proved by either
	1) $\leafbeta{q}$ for some $q : Q \; a$, hence with $p : P \; a$ and disjointness of $P$ and $Q$, we get a proof of absurdum;
	or 2) $\exepbeta{b'}$ for some $b' : \pileaf{o} \equiv \falsebool$, hence  $\truebool \equiv \falsebool$ by transitivity and symmetry applied to $b$ and $b'$.
	
	If $t = \node{k}{ts}$, proofs of $\liftTree{o}{P}{t}$ and $\coliftTree{o}{Q}{t}$ are required to be of the form $\nodealpha{ps}$ and $\nodebeta{qs}$ with $ps : \liftTest {(\lambda (x,o').\;\alpha\;o'\;(Ps\;x))} \; (\pinode{k}{o})$ and $qs : \liftTest {(\lambda (x,o').\;\beta\;o'\;(Ps\;x))} {(\dualTest{(\pinode{k}{o})})}$. By inductive hypothesis, the predicates $\lambda (i, o'). \;\liftTree{o'}{P}{(ts \; i)}$ and $\lambda (i, o'). \;\coliftTree{o'}{Q}{(ts \; i)}$ are disjoint. And this, together with the presence of both proof terms $ps$ and $qs$, is absurd by Lemma \ref{lem:distest}.
\end{proof}
Notice that the above proof is finite, in the sense that the term of
inductive type $\liftTree{o}{P}{t}$ gets smaller in each recursive call.

\subsection{Examples}

Let us now look at our running examples. Remember that we dualize the tests of our observations.

\begin{example}\label{ex:pure-co}
	Consider Example \ref{ex:pure1} concerning pure computations using the $\skipo$ operation.
	In case $\skipo$ is undetectable, and we have only one observation ${\downarrow}$ for termination,
	$\coliftTree{{\downarrow}}{P}{t}$ expresses partial correctness of $t$: either $t$ terminates producing a result correct under $P$, or $t$ diverges producing infinitely many skips.
	We get that $\liftTree{\downarrow}{P}{t}$ implies $\coliftTree{\downarrow}{P}{t}$.
	In the case that the $\skipo$s can be counted and $O = \Nat$, $\coliftTree{n}{P}{t}$ expresses partial correctness of $t$ under time limit $n$: either $t$ terminates within $n$ skips producing a result correct under $P$, or $t$ takes at least $n+1$ skips. We get that $\liftTree{n}{P}{t}$ implies $\coliftTree{n}{P}{t}$.
\end{example}

\begin{example}\label{ex:nond-co}
	Consider Example \ref{ex:nond1} concerning nondeterministic computations using the binary choice operation $\Or$.
	Then $\coliftTree{\may}{P}{t}$ states that it is not possible to get a result which does not satisfy $P$, so any decision process for $t$ must either lead to divergence or the production of a result correct under $P$.
	On the other hand, $\coliftTree{\must}{P}{t}$ says that we cannot guarantee that $t$ produces a result which does not satisfy $P$, so either $t$ may diverge or it may produce a result correct under $P$.
	
	Note that $\dualTestS$ swaps the $\vee$ and $\wedge$ definitions in $\pinode{\Or}{o}$.
	As such, $\beta \; \may$ acts like a partial correctness version of must, and $\beta \; \must$ acts like a partial correctness version of may.
	As a result, $\liftTree{\may}{P}{t}$ implies $\coliftTree{\must}{P}{t}$, and $\liftTree{\must}{P}{t}$ implies $\coliftTree{\may}{P}{t}$.
\end{example}

\begin{example}\label{ex:glob-co}
	We already briefly discussed the coinductive predicate lifting generated by global store observations from Example \ref{ex:glob1}.
	Then $\coliftTree{(n,m)}{P}{t}$ states that: if with starting state $n$ the program $t$ terminates with final state $m$, then it produces a result correct under $P$.
	We have that $\liftTree{(n,m)}{P}{t}$ implies $\coliftTree{(n,m)}{P}{t}$.
	Moreover, due to the exception base case $\exepbetaS$, we have that for $m \neq k$, $\liftTree{(n,k)}{(\lambda x.1)}{t}$ implies $\coliftTree{(n,m)}{P}{t}$.
	These predicates vary from more traditional partial correctness properties for global store from the literature. However, with the logic in Section \ref{sec:logic}, these alternative formulations can be reconstructed.
\end{example}

\begin{example}\label{ex:inpu-co}
	Lastly, we consider the input requesting computations from Example \ref{ex:inpu1}.
	These use as observations $O = 2 \times 2^*$, a pair consisting of a bit and a list of bit inputs.
	Then $\coliftTree{(\leftbool,l)}{P}{t}$ expresses the partial property telling us that: 
	if $t$ can be given the input sequence $l$ and if $t$ terminates after it is given, then it returns a value satisfying predicate $P$.
	On the other hand, $\coliftTree{(\rightbool,l)}{P}{t}$ says:
	if the user is able to input the sequence $l$, then the computation $t$ will not ask for another input afterwards.
	Note that neither $\liftTree{(\rightbool,l)}{P}{t}$ nor $\coliftTree{(\rightbool,l)}{P}{t}$ concern themselves with the predicate $P$.
\end{example}

\noindent
\textbf{A Note on Computability.}
The coinductive counterpart $\coliftTree{o}{P}{t}$ to the inductive property $\liftTree{o}{Q}{t}$ does not offer a complete notion of negation, such as is offered by the more usual functions into the empty datatype $\liftTree{o}{P}{t} \to 0$.
It is however a large constructive distinct property, which allows for double negation elimination.
We briefly explore the difference between the two notions of negations.

Consider a possible enumeration of pairs of all Turing machines and their input arguments, and a function returning $\ttt : 1$ on the ones that terminate on their input. In other words, we consider a function $\mathsf{TM} : \Nat \to \Tree{I}{1}$, where $I$ is the signature with skips of Example~\ref{ex:pure1}.
Consider the always true and always false predicates $T, F : 1 \to \Set$. 
One can construct a function $\liftTree{\downarrow}{T}{\mathsf{TM}} : \Nat \to \Set$ which for $n : \Nat$ gives proofs that the $n$-th Turing Machine terminates on its input, and also a function $\coliftTree{\downarrow}{F}{\mathsf{TM}} : \Nat \to \Set$ which for $n : \Nat$ gives proofs that the $n$-th Turing Machine diverges.
The two predicates on $\Nat$ are necessarily disjoint.
Moreover, by the Halting problem, they cannot partition $\Nat$ either.
Hence, they are not complete negations of each other.

This program can also be adapted to show differences between distinctions for other examples. In nondeterminism, $\lambda n.\; \Or(\mathsf{TM} \ n, \leaf{\ttt})$ lies in the gap between $\liftTreeS{\must}\;{T}$ and $\coliftTreeS{\must}\;{F}$.
The program $\lambda n. \;\Or(\mathsf{TM} \ n, \Omega)$, where $\Omega$ is a provably always diverging computation, lies in the gap between $\liftTreeS{\may}\;{T}$ and $\coliftTreeS{\may}\;{F}$.

\section{Behavioural Logic}\label{sec:logic}

In this section we introduce our generic programming language of denotations and
formulate a logic for expressing behavioural properties of programs in
this language.
These are meta-theoretic programs as can be given in Agda terms, rather than one specific programming language.
As such, it is similar to other generic languages based on coinductive trees, like those formulated around interaction trees \cite{XiaZHHMPZ20}.
Their type is specified by the following small grammar, also appearing in Moggi's monadic metalanguage~\cite{monad}:
\begin{center}
$\sigma,\rho ::= \N \, | \, \sigma \tto \rho \, | \, \sigma \otimes \rho\,|\, \U\;\sigma$. 
\end{center}

The collection of these syntactic types is called $\Ty$, and it
includes names for the type of natural numbers,
function type, Cartesian product and a unary type former $\U$ for turning
computations into values, often present in call-by-value
languages. Well-typed terms of syntactic type $\sigma$ are elements of
type $\Tm \;b \;\sigma$, where $b : \{\val,\cpt\}$ is a Boolean used to
distinguish value terms from computation terms.
Notably, computation terms of syntactic type $\sigma$
are coinductive trees returning $\sigma$-typed values in their
leaves. Value terms of syntactic type $\U\;\sigma$ are exactly
computation terms of syntactic type $\sigma$.

\begin{center}{\small
		\begin{tabular}{l l | l | l}\label{eq:terms}
			$\Val : \Ty \to \Set$	&	& $\Cpt : \Ty \to \Set$ & $\Tm : \Bool \to \Ty \to \Set$ \\
			\hline
			$\Val\;\N = \Nat$ & $\Val\;(\sigma \otimes \rho) = \Val\;\sigma \times \Val\;\rho$ & $\Cpt\;\sigma = \Tree I {(\Val\;\sigma)}$ &  $\Tm\;\val\;\sigma = \Val\;\sigma$ \\
			$\Val\;(\U\;\sigma) = \Cpt\;\sigma$ & $\Val\;(\sigma \tto \rho) = \Val\;\sigma \to \Cpt\;\rho$ &  & $\Tm\;\cpt\;\sigma = \Cpt\;\sigma$
	\end{tabular}}
\end{center}

Our behavioural logic is composed of value formulae and computation
formulae. Formulae are elements of type $\Formo b \sigma$, where $b$ is
either $\val$ (value) or $\cpt$ (computation) and $\sigma$ is a
syntactic type, which are inductively generated by the following
inference rules:
\begin{equation}\label{eq:formulae}
	\def\arraystretch{2.3}  
	\begin{array}{c}
	\infer{\eq\;n : \Formo \val \N}{n : \Nat}\qquad
	\infer{\neqF\;n : \Formo \val \N}{n : \Nat}\qquad
	\infer{V \mapsto \phi : \Formo \val {(\sigma \tto \rho)}}
	{V : \Val\;\sigma & \phi : \Formo \cpt \rho}\\
	\infer{\lef \;\phi : \Formo \val {(\sigma \otimes \rho)}}
	{\phi : \Formo \val \sigma}\qquad
	\infer{\rig \;\phi : \Formo \val {(\sigma \otimes \rho)}}
	{\phi : \Formo \val \rho}\qquad
	\infer{\thunk\;\phi : \Formo \val {(\U\;\sigma)}}
	{\phi : \Formo\cpt\sigma}\\
	\infer{\test\; \phi s : \Formo b \sigma}
	{\phi s : \Test {(\Formo b \sigma)}}\qquad
	\infer{\obsalpha\; o\; \phi : \Formo \cpt \sigma}
	{o : O & \phi : \Formo \val \sigma}\qquad
	\infer{\obsbeta\; o\; \phi : \Formo \cpt \sigma}
	{o : O & \phi : \Formo \val \sigma}
	\end{array}
\end{equation}

\noindent
We assume $(\liftTreeS{(-)}, \coliftTreeS{(-)})$ is a complementing pair of $O$-indexed predicate liftings. We define satisfiability as a relation between syntactic terms in $\Tm \;
b\;\sigma$ and formulae in $\Formo b \sigma$.
\begin{center}
	$
	\arraycolsep=1.5pt
	\begin{array}{lcl@{\qquad}lcl@{\qquad}lcl}
	m \models \eq\;n & = & m \equiv n &
	(V , W) \models \lef \;\phi & = & V \models \phi &
	P \models \test\;\phi s & = & \liftTest {(\lambda \phi.\;P\models \phi)} \phi s \\
	m \models \neqF\;n & = & m \not\equiv n &
	(V , W) \models \rig \;\phi & = & W \models \phi &
	M \models \obsalpha\;o\;\phi & = &\liftTree {o} (\lambda V.\;V \models \phi) \ M \\
	W \models V \mapsto \phi & = & W\;V \models \phi &
	V \models \thunk\;\phi & = & V \models \phi &
	M \models \obsbeta\;o\;\phi & = & \coliftTree {o} (\lambda V.\;V \models \phi) \ M 
	\end{array}
	$
\end{center}

The formulation of the logic and the satisfiability relation is quite
standard~\cite{modal-journal}, with a few exceptions. Formula formers
$\obsalpha\;o$ and $\obsbeta\;o$ play the role of modalities. E.g. in
the case of Example \ref{ex:nond1}, $\obsalpha\;\may$ and
$\obsalpha\;\must$ correspond to modalities $\Diamond$ and $\Box$, and
$\obsbeta\;\may$ and $\obsbeta\;\must$ are coinductive variants taking
into account possible non-termination. The formula former $\test$
allows to construct formulae via the simple test logic $\Test$.
For function testing, we opted for a concrete argument passing test as used in \cite{modal-journal}. This follows traditional testing approaches as used in applicative bisimilarity \cite{Abramsky90} and corresponding testing logics.
Alternatively, Hoare-style predicates based on preconditions could be utilized. However, it seems in practice these are both cumbersome and difficult to work with.

The logic can be used for testing a variety of useful properties. Particularly interesting is the nesting of inductive and coinductive predicates, obtained by putting the predicates in a sequence.
Using combinations of predicate liftings, we can easily swap between inductive and coinductive at different type levels, as needed.
% Previously Outcommented example:
As an example, using countable skips from Examples \ref{ex:pure1} and \ref{ex:pure-co}, we can formulate a formula capturing the following property of programs of type $\U \, \sigma$: if we evaluate both the program and the result it produces, it will not take longer to evaluate the result then it took to evaluate the initial program. This can for instance be constructed as follows: $\test (\bigwedge \lambda n . \ \atom (\obsbeta \ n \ (\thunk \ (\obsalpha \ n \ (\test \ \true)))))$.

Satisfiability of formulas is employed in the specification of an extensional
ordering on syntactic terms. Given $P,Q : \Tm\;b\;\sigma$, we define the
\emph{logical approximation} $P \le_\Tm Q = (\phi : \Formo b \sigma)
\to P \models \phi \to Q \models \phi$. Program $P$ is below program
$Q$ in this ordering if and only if every formula satisfied by $P$ is also
satisfied by $Q$. An extensional \emph{logical equivalence} of
syntactic terms $P \equiv_\Tm Q$ is given by $(P \le_\Tm Q) \times (Q \le_\Tm P)$.

The negation of a formula $\phi$ is not among the generating
connectives of the behavioural logic in
(\ref{eq:formulae}). Nevertheless, an emergent notion of negation is
admissible.

\begin{definition}
	We define negation as a function $\negFormS : \Formo b \sigma \to \Formo b \sigma$  using the following rules:
	\begin{center}
		$
			\arraycolsep=1.4pt
			\begin{array}{lcl@{\quad \ \ }lcl@{\quad \ \ }lcl}
			\negForm {(\eq\;n)} & = & \neqF\;n &
			\negForm {(\neqF\;n)} & = & \eq\;n \\
			\negForm {(\lef\; \phi)} & = & \lef \;(\negForm \phi) &
			\negForm {(\rig\; \phi)} & = & \rig \;(\negForm \phi)\\
			\negForm {(V \mapsto \phi)} & = & V \mapsto \negForm \phi &
			\negForm {(\thunk\; \phi)} & = & \thunk \;(\negForm \phi) \\
			\negForm {(\obsalpha\;o\;\phi)} & = & \obsbeta\;o\;(\negForm \phi) &
			\negForm {(\obsbeta\;o\;\phi)} & = & \obsalpha\;o\;(\negForm \phi)\\
			\negForm {(\test\; \phi s)} & = & \test\;(\mapTest {\negFormS} {\phi s}) & & &
			\end{array}
		$
	\end{center}
\end{definition}
The base cases $\eq$ and $\neqF$ are each others
complements, and a similar relationship exists between $\obsalpha$ and
$\obsbeta$.
Negation $\negFormS$ is involutive in the following sense: for all
formulae $\phi$, we have the equivalence $\negForm{(\negForm\phi)} \equiv \phi$. Notice in
particular that the double negation of a formula $\phi$ is
syntactically equal to $\phi$, not merely related by some notion of
logical equivalence.

The use of the word ``negation'' for the formula operation
$\negFormS$ is justified by the fact that no syntactic program $P$
satisfies simultaneously a formula $\phi$ and its complement $\negForm
\phi$.

\begin{proposition}\label{prop:disj-form}
	For all $\phi : \Formo b \sigma$, the predicates $(-) \models \phi$ and $(-)
	\models \negForm \phi$ are disjoint.
\end{proposition}
\begin{proof}
	Given a term $P$, we need to show that assuming both $P \models \phi$
	and $P \models \negForm \phi$ is absurd. The interesting case is $\phi
	= \obsalpha\;o\;\psi$ (and dually $\phi = \obsbeta\;o\;\psi$, which is
	proved in a similar way). In this case, the two assumptions rewrite to
	$\liftTree {o} {(\lambda V.\;V \models \psi)} P$ and $\coliftTree
	{o} {(\lambda V.\;V \models \negForm \psi)} P$ respectively. By
	inductive hypothesis, $(-) \models \psi$ and $(-) \models \negForm
	\psi$ are disjoint, therefore invoking
	Proposition~\ref{prop:dist-pred} generates a contradiction.
\end{proof}

The logical approximation $\le_\Tm$ is not symmetric, since $P \le_\Tm
Q$ does not generally entail $Q \le_\Tm P$. The validity of $P \le_\Tm
Q$ seems also not sufficient for deriving ${(\phi : \Formo b
	\sigma) \to Q \models \phi} \to ((P \models \phi) \to 0) \to 0$, i.e. a
doubly-negated variant of $Q \le_\Tm P$. What can be proved from $P
\le_\Tm Q$ instead is the impossibility of $P$ and $Q$ to satisfy dual
formulae: for any formula $\phi$, it is not the case that both $Q \models \phi$ and $P \models
\negForm \phi$ are derivable. This suggests
the introduction of a weak notion of logical approximation $P
\simeq_\Tm Q = (\phi : \Formo b \sigma) \to (P \models \phi) \to (Q \models \negForm \phi) \to 0$, so that $P \le_\Tm Q$ implies $P
\simeq_\Tm Q$. Unlike the logical approximation $\le_\Tm$, the
relation $\simeq_\Tm$ is symmetric, which is easily provable invoking
the involutive property of negation $\negFormS$.

Alternatively, we could consider a logic with only inductive predicate liftings like in \cite{modal-journal}, and a logic with only coinductive predicate liftings. The negation operation is a function between the two logics, which sends a formula of one logic to its constructive complement.
Even if one only considers behavioural equivalence under one of these logics, statements like $M \models \phi \to N \models \phi$ may be verified or falsified using the constructive negation of $\phi$ from the other logic.
As such, coinductive statements enrich our toolkit for reasoning about pre-existing notions of equivalence based on induction.

\section{Decomposability}\label{sec:compos}
Higher-order computations can return values which are computations themselves.
Using the logic of the previous section, we can formulate behavioural properties of such computations.
Properties of higher-order programs may contain multiple inductive and/or coinductive predicate liftings in succession, allowing for nuanced statements.

One thing we can do with a higher-order program is \emph{sequence} it.
This evaluates the program, and then continues by evaluating its result, thereby putting the two evaluations in sequence.
The question is, can we prove whether a sequenced program satisfies some behavioural property, using behavioural properties of the unsequenced higher-order program?
Answering this question is fundamental for establishing a plethora of compositionality results.

To simplify the question, we consider \emph{double trees} of type $\Tree{I}{(\Tree{I}{\Set})}$.
These trees occur naturally as a result of, for instance, the study of computations of type $\U \; \tau$ using a value formula on $\tau$.
The sequencing of programs corresponds to an application of the monad multiplication map $\mu~:~\Tree{I}{(\Tree{I}{\Set})} \to \Tree{I}{\Set}$.
Both single trees in $\Tree{I}{\Set}$ and double trees have natural relations of extensional ordering:

\begin{itemize}
	\item For $t_0, t_1 : \Tree{I}{\Set}$ we say $t_0 \sqsubseteq t_1$ if for any $o : O$, 
	$\alpha \; o \; t_0$ implies $\alpha \; o \; t_1$, and $\beta \; o \; t_0$ implies $\beta \; o \; t_1$.
	\item For $d_0, d_1 : \Tree{I}{(\Tree{I}{\Set})}$ we say $d_0 \sqsubseteq d_1$ if for any $o_0, o_1 : O$, 
	$\liftTree{o_0}{(\alpha \; o_1)}{d_0}$ implies $\liftTree{o_0}{(\alpha \; o_1)}{d_1}$, and $\coliftTree{o_0}{(\beta \; o_1)}{d_0}$ implies $\coliftTree{o_0}{(\beta \; o_1)}{d_1}$.
\end{itemize}

Statements of the form $\liftTree{o_0}{(\alpha \; o_1)}{d_0}$ and $\coliftTree{o_0}{(\beta \; o_1)}{d_0}$ are called \emph{observational towers}, hinting at the fact that these can be extended to include even more layers of predicate liftings.
We express preservation of observations over program sequencing as follows.

\begin{definition}
	We call a triple $(O, \pileafS, \pinodeS)$ \emph{strongly decomposable} if for any two double trees $d_0, d_1 : \Tree{I}{(\Tree{I}{\Set})}$, $d_0 \sqsubseteq d_1$ implies $\mu \; d_0 \sqsubseteq \mu \; d_1$. 
\end{definition}

One way of establishing strong decomposability is by showing that any observations $\alpha \; o$ and $\beta \; o$ on sequenced programs can be \emph{decomposed} into a test $\pid \; o : \Test{(O \times O)}$ whose atoms are given by pairs $(o_0,o_1) \in O \times O$ representing observational tower statements. 

\begin{definition}
	We say that $\pid : O \to \Test{(O \times O)}$ is an \emph{$\alpha$-decomposition} if, for any $d : \Tree{I}{(\Tree{I}{\Set})}$, $\alpha \; o \; (\mu \; d)$ iff $\liftTest{(\lambda (o_0, o_1). \;\liftTree{o_0}{(\alpha \; o_1)}{d})}{(\pid \; o)}$.
	We say that $\pid : O \to \Test{(O \times O)}$ is a \emph{$\beta$-decomposition} if,
	for any $d : \Tree{I}{(\Tree{I}{\Set})}$, $\beta \; o \; (\mu \; d)$ if and only if $\liftTest{(\lambda (o_0, o_1). \;\coliftTree{o_0}{(\beta \; o_1)}{d})}{(\pid \; o)}$.
\end{definition}

\begin{lemma}\label{lem:decom}
	$(O, \pileafS, \pinodeS)$ is strongly decomposable if there exist both an $\alpha$-decomposition and a $\beta$-decomposition.
\end{lemma}

\noindent
Note that in the lemma, we do not require the two decompositions to be the same. In fact, in our examples described later on, it is possible to find an $\alpha$-decomposition $\pid$ such that its dual $\lambda o.\; \dualTest{(\pid \; o)}$ is a $\beta$-decomposition.

The word `strong' in `strongly decomposable' reflects the fact that the relation $\sqsubseteq$ between double trees is weaker than the relation induced by the logic on higher-order programs.
As a result, the property is stronger than necessary for establishing most compositionality results.
To establish the weaker notion of sequence preservation, \emph{decomposability}, it is sufficient to show the existence of $\alpha$ and $\beta$-decompositions of the form $O \to \Test{(O \times \Test{O})}$.
Luckily, all examples we consider are strongly decomposable following Lemma \ref{lem:decom}.
Hence we need not consider this weaker notion.

The decompositions can be generalized to any computation of type $\U \;\tau$, for any syntactic type $\tau : \Ty$. Given a formula $\phi : \Formo{\Val}{\tau}$ and a test $t : \Test{(O \times O)}$, we define a formula $t_{\alpha}[\phi] : \Formo{\Cpt}{(\U \; \tau)}$ 
as $\test\;(\mapTest{(\lambda (o_1, o_2).\; \obsalpha \; o_1 \; (\thunk \; \obsalpha \; o_2 \; \phi))}{t})$. We define a formula $t_{\beta}[\phi] : \Formo{\Cpt}{(\U \; \tau)}$ in the same way but replacing each occurrence of $\alpha$ by $\beta$.

\begin{lemma}
	Suppose given $\pid : O \to \Test{(O \times O)}$,
	\begin{itemize}
		\item if $\pid$ is an $\alpha$-decomposition, then, for any $\phi : \Formo{\val}{\tau}$, $M : \Cpt \; \tau$ and $o : O$, we have
		$\mu \; M \models \obsalpha \; o \; \phi$ if and only if $M \models (\pid \; o)_{\alpha}[\phi]$;
		\item if $\pid$ is a $\beta$-decomposition, then, for any $\phi : \Formo{\val}{\tau}$, $M : \Cpt \; \tau$ and $o : O$, we have
		$\mu \; M \models \obsbeta \; o \; \phi$ if and only if $M \models (\pi d \; o)_{\beta}[\phi]$.
	\end{itemize}
\end{lemma}

% Outcommented example
%A more concrete example considers a traditional transformation on functional programs: the currying operation $\mathsf{C}$, sending values of type $(\sigma_0 \otimes \sigma_1) \tto \tau$ to values of type $\sigma_0 \tto (\sigma_1 \tto \tau)$ as follows: $\mathsf{C}\;f = \lambda V_0. \; \eta\;(\lambda V_1. \;f\;(V_0, V_1)).
%$
%The uncurrying operation $\mathsf{C'} : \Val\;(\sigma_0 \tto (\sigma_1 \tto \tau)) \to \Val \;((\sigma_0 \otimes \sigma_1) \tto \tau)$ is given by:
%$\mathsf{C'}\;{g} = \lambda (V_0, V_1). \; \mu \; (\mapTree{(\lambda h. \; h\;V_1)}{(g\;V_0)})$.
%Note that $\mathsf{C'}\;(\mathsf{C}\;f) = f$. However, $\mathsf{C}\;(\mathsf{C'}\;g)$ is in general not equal to $g$.
%This is due to the fact that at type $\sigma_0 \tto (\sigma_1 \tto \tau)$, effects may arise at two separate instances: after the first argument is given and after the second argument is given.
%Uncurrying via $\mathsf{C'}$ merges these two moments into one.
%This is an example of sequencing occurring naturally within standard program manipulations.
%
%\begin{proposition}
%	The currying transformation $\mathsf{C}$ preserves logical approximation $\le_\Tm$.
%	Moreover, if $(O, \pileaf{\!}{\!}, \pinode{\!}{\!}{\!\!})$ is strongly decomposable then the uncurrying transformation $\mathsf{C'}$ preserves logical approximation.
%\end{proposition}

We present instances of $\alpha$ and $\beta$-decompositions for our examples, proving they are strongly decomposable.
In each instance, the $\beta$-decomposition is given by the test dual to the $\alpha$-decomposition.

\begin{example}\label{ex:pure-dec}
	Consider Example \ref{ex:pure1} concerning pure computations using an undetectable $\skipo$ operation.
	Both the $\alpha$- and $\beta$-decomposition can be given by $\pid \; {\downarrow} = \atom{({\downarrow}, {\downarrow})}$ signifying that
	a sequenced program terminates if and only if the unsequenced program terminates and returns a terminating program.
	
	In the case that the number of $\skipo$'s can be measured, we give the $\alpha$-decomposition as 
	
	\begin{center}
		$\pid \; n = \bigvee \left( \lambda m. \; \bigvee \left(\lambda k. \; \text{if} \; (m + k \leq n) \; \text{then} \; \atom{(m, k)} \; \text{else} \; \false\right)\right)$.
	\end{center}
	
	\noindent
	This states that a sequenced program terminates within $n$ steps, if the unsequenced program returns within $m$ steps a program which itself terminates within $k$ steps, such that the total number of steps $m + k$ is at most $n$.
	The $\beta$-decomposition $\pid'$ is given by the dual:
	\begin{center}
	$\pid' \; n = \bigwedge \left( \lambda m. \; \bigwedge \left(\lambda k. \; \text{if} \; (m + k \leq n) \; \text{then} \; \atom{(m, k)} \; \text{else} \; \true\right)\right)$.
	\end{center}
\end{example}

\begin{example}\label{ex:nond-dec}
	For nondeterministic computations following Example \ref{ex:nond1},
	both $\alpha$- and $\beta$-decompositions are given by $\lambda o. \; \atom{(o, o)}$. E.g., a sequenced program may terminate if the unsequenced program may return a program which may terminate.
	
\end{example}

\begin{example}\label{ex:glob-dec}
	For global store computations from Example \ref{ex:glob1},
	the decompositions check for the existence of some intermediate value which specifies the state after the first evaluation and before the second evaluation of a double tree.
	The $\alpha$-decomposition is given by
	$\pid \; (n, m) = \bigvee \left(\lambda k. \; \atom{((n, k), (k, m))}\right)$.
	Intuitively, this states that a sequenced program goes from state $n$ to terminating with state $m$, if and only if an unsequenced program goes from $n$ to some intermediate state $k$, and evaluating the returned program with state $k$ yields termination with state $m$.
	The $\beta$-decomposition is given by the dual $\pid' \; (n, m) = \bigwedge \left(\lambda k. \; \atom{((n, k), (k, m))}\right)$.
\end{example}

\begin{example}\label{ex:inpu-dec}
	The decomposition of input observations from Example \ref{ex:inpu1} relies on the enumerability of the set of bitlists $2^*$ and the decidability of its propositional equality. 
	Notice that countable conjunction $\bigwedge$ and disjunction $\bigvee$ in $\Test$ are both indexed by natural numbers. But, due to to the enumerability of bitlists, they can also be indexed by $2^*$.
	We define an $\alpha$-decomposition $\pid : O \to \Test{(O \times O)}$ for the input example, whose dual is a $\beta$-decomposition. This function splits the bitlist $l$ into a bitlist $x$ before the first termination, and a bitlist $y$ for testing the program after the first termination.
	
	\begin{center}
		$   
		\arraycolsep=1pt
		\begin{array}{rcl}  
		\pid \; (\leftbool, l) &=& \bigvee \left( \lambda x , y : 2^*. \; \text{if} \; (x{+}\!\!\!{+} y \equiv l) \; \text{then} \; \atom{((\leftbool, x), (\leftbool, y))} \; \text{else} \; \false \right) \\	
		\pid \; (\rightbool, l) &=& \atom{((\rightbool, l), (\rightbool, l))} \;\vee \\
		& & \bigvee \left(\lambda x , y : 2^*. \; \text{if} \; (x{+}\!\!\!{+} y \equiv l) \; \text{then} \; \atom{((\leftbool, x), (\rightbool, y))} \; \text{else} \; \false\right)
		\end{array}
		$
	\end{center}
\end{example}

\noindent
\textbf{Higher-order Nondeterminism.} 
We study some example of nondeterministic programs.
Firstly, for any type $\tau$, we have the diverging computation $\Omega_\tau : \Cpt \; \tau$ defined corecursively as the program satisfying the following equation: $\Omega_\tau = \node {\Or} {(\lambda \,\leftbool.\;\Omega_\tau, \lambda \,\rightbool.\;\Omega_\tau)}$. For readability's sake, we simply write $\Or$ in place of $\nodeS\;\Or$, and we avoid writing $\lambda \,\leftbool$ and $\lambda \,\rightbool$, since it is clear that the left argument corresponds to the left branch of the tree, similarly for the right case.
It can be shown that the $\Omega_\tau$ computation \emph{must diverge}, satisfying the formula $\mathsf{mustdiv} = \negForm{(\obsalpha \; \may \; (\test \; \true))}$, which is equal to $\obsbeta \; \may \; (\test \; \false)$.
Note the use of the formula $\test \; \false$ which cannot be satisfied, hence $\mathsf{mustdiv}$ cannot be satisfied by a computation that can produce a result.

As a case study of properties of higher-order programs, we consider the following two computations of type $\U \;(\U \;\tau)$ for some type $\tau$, a variation on an example from \cite{Lassen,Ong}: Take the program
$ P = \Or \;(\leaf{(\Omega_{\U \tau})}, \; \leaf{(\leaf{\Omega_\tau})})$ and $Q= \leaf{(\Or \;(\Omega_{\U \tau}, \; \leaf{\Omega_\tau}))} $.
We can show that $P$ may produce a result that must diverge, and hence it satisfies $\phi = \obsalpha \; \may \; (\thunk \;\mathsf{mustdiv})$.
On the other hand, all programs which $Q$ returns may terminate, hence $Q$ satisfies $\negForm{\phi} = \obsbeta \; \may \; (\thunk \; (\obsalpha \;\may \; (\test \; \true)))$.
Hence, $P$ and $Q$ are not related via the logical equivalence.
Note moreover that $\phi$ only uses the $\may$ observation.
Via the connection to applicative bisimilarity shown in the next section, we are able to reproduce Ong's result that these two terms are not applicatively bisimilar~\cite{Ong}. See \cite{VoorneveldThesis} for more classical examples of properties for higher-order effectful programs.

\section{Simulations and Equality}

The logic from section \ref{sec:logic} is capable of expressing a behavioural or observational difference between two programs.
For two terms $P, Q : \Tm \; b \; \tau$, such a difference is a formula $\phi : \Formo{b}{\tau}$ such that $P \models \phi$ and $Q \models \negForm{\phi}$.
This shows that $P \not\equiv_\Tm Q$.
Showing that two programs have no behavioural differences is not trivial.
One strategy for proving behavioural equivalence is using \emph{applicative simulations}~\cite{Abramsky90}.

Applicative simulations, or more generally just simulations, are relations on program terms which prove \emph{similarity} between them.
If two programs are related by a simulation they are \emph{similar}, and if they are related by a symmetric simulation, also called a \emph{bisimulation}, then they are \emph{bisimilar}.
In this section we define these notions, and show how they relate to our logically induced behavioural equivalence.

In the context of algebraic effects, we need to specify how to compare the effect operations of two programs.
This is done in \cite{Relational} using \emph{relators} \cite{Levy11}, which lift relations between types $X$ and $Y$ to relations between the monadic liftings $\Tree{I}{X}$ and $\Tree{I}{Y}$.
In this paper, we use a simpler variation of the notion of relator considering only homogeneous relations.
Given some relation $\mathcal{R} : A \to A \to \Set$, a predicate $f : A \to \Set$ is considered \emph{$\mathcal{R}$-correct}
if there is a function of type $(a \; b : A) \to f \; a \to \mathcal{R} \; a \; b \to f \; b$.

\begin{definition}\label{def:our-rel-lift}
	Given $(O, \pileafS, \pinodeS)$ and $\mathcal{R} : A \to A \to \Set$, we define a new relation $\Gamma(\mathcal{R}) : \Tree{I}{A} \to \Tree{I}{A} \to \Set$ such that $\Gamma(\mathcal{R}) \; t_0 \; t_1$ holds if and only if, for any $\mathcal{R}$-correct predicate $f$ and observation $o:O$, we have the following two implications:\quad
	$\liftTree{o}{f}{t_0} \to \liftTree{o}{f}{t_1}$ \ \ and \ \      
	$\coliftTree{o}{f}{t_0} \to \coliftTree{o}{f}{t_1}$ .
\end{definition}
The original definition of a relator using modalities involves a feedback loop between predicates and relations, which necessarily increases the universe level when describing higher order types. This is problematic when trying to show preservation over sequencing. The above variation avoids this problem, at the cost of compositionality, but without sacrificing transitivity of the resulting notions of program equivalence.

The relation lifting satisfies some relator properties, such as monotonicity with respect to the relation order $\mathcal{R} \subseteq \mathcal{S}$.
Moreover, the resulting relation behaves well with respect to sequencing, as seen below:

\begin{proposition}\label{prop:relseq}
	If $(O, \pileafS, \pinodeS)$ is strongly decomposable, then for any relation ${\mathcal{R} : A \to A \to \Set}$ and any two double trees $d_0, d_1 : \Tree{I}{(\Tree{I}{A})}$, if
	$\Gamma(\Gamma(\mathcal{R})) \; d_0 \; d_1$ then $\Gamma(\mathcal{R}) \; (\mu \; d_0) \; (\mu \; d_1)$.
\end{proposition}

The relation lifting $\Gamma$ is unfortunately not \emph{compositional}, since it does not preserve relation composition.
Despite this fact though, we are still able to show transitivity of the resulting notions of program relation, as we will see in Proposition \ref{prop:sim-trans}.

Using the relation lifting $\Gamma$, we define a notion of similarity between programs.
A \emph{well-typed} relation $\overline{\mathcal{R}}$ is a collection of relations on syntactic terms indexed by sorts and types:
$\overline{\mathcal{R}} : {(b : \{\val, \cpt\})} \,(\tau : \Ty) \to \Tm \; b \; \tau \to \Tm \; b \; \tau \to \Set$.

\begin{definition}
	A well-typed relation $\overline{\mathcal{R}}$ is an \emph{applicative $\Gamma$-simulation} when:
	\begin{itemize}
		\item If \, $\overline{\mathcal{R}} \; \val \; \Nat \; n \; m$, then $n \equiv m$.
		\item If \, $\overline{\mathcal{R}} \; \val \; (\sigma \tto \tau) \; V \; W$, then, for any $U : \Val \; \sigma$,  \, $\overline{\mathcal{R}} \; \cpt \; \tau \; (V\,U) \; (W\,U)$.
		\item If \, $\overline{\mathcal{R}} \; \val \; (\sigma \otimes \tau) \; (V_0, V_1) \; (W_0, W_1)$, \, then $\overline{\mathcal{R}} \; \val \; \sigma \; V_0 \; W_0$ and $\overline{\mathcal{R}} \; \val \; \tau \; V_1 \; W_1$.
		\item If \, $\overline{\mathcal{R}} \; \val \; (\U \; \tau) \; V \; W$, then $\overline{\mathcal{R}} \; \cpt \; \tau \; V \; W$.
		\item If \, $\overline{\mathcal{R}} \; \cpt \; \tau \; M \; N$, then $\Gamma(\overline{\mathcal{R}} \; \val \; \tau) \; M \; N$.
	\end{itemize}
\end{definition}

\begin{definition}
	For any sort $b : \{\val, \cpt\}$ and type $\tau : \Ty$, we call two terms $P, Q : \Tm \; b \; \tau$ \emph{applicatively $\Gamma$-similar} if there is a $\Gamma$-simulation $\overline{\mathcal{R}}$ such that $\overline{\mathcal{R}} \; b \; \tau \; P \; Q$ holds.
	Two terms $P, Q : \Tm \; b \; \tau$ are called \emph{applicatively $\Gamma$-bisimilar} if there is a symmetric $\Gamma$-simulation $\overline{\mathcal{R}}$ such that $\overline{\mathcal{R}} \; b \; \tau \; P \; Q$ holds.
\end{definition}

Importantly, the notions of similarity and bisimilarity give us proof techniques for showing that there are no observable differences between programs.

\begin{proposition}
	Given two programs $P$ and $Q$ of the same type, then:
	\begin{itemize}
		\item If $P$ and $Q$ are applicatively $\Gamma$-similar, then $P \leq_\Tm Q$.
		\item If $P$ and $Q$ are applicatively $\Gamma$-bisimilar, then $P \equiv_\Tm Q$.
	\end{itemize} 
\end{proposition}
\begin{proof}
	Given an applicative $\Gamma$-simulation $\overline{\mathcal{R}}$, we show by induction on formulae $\phi$, that for any $P$ and $Q$ related by $\overline{R}$, $P \models \phi$ implies $Q \models \phi$. If this holds, we say $\overline{\mathcal{R}}$ preserves $\phi$.
	
	Most cases are straightforward. Most difficult are the instances where $\phi = \obsalpha \; o \; \psi$ and $\phi = \obsbeta \; o \; \psi$ of some type $\tau$.
	In these cases, we use the induction hypothesis on $\psi$ to show $\overline{\mathcal{R}}$ preserves $\psi$.
	Hence $\lambda V. (V \models \psi)$ is $(\overline{\mathcal{R}} \; \val \; \tau)$-correct, and the result can be derived from the appropriate simulation condition.
\end{proof}

It seems however impossible to derive implications in the opposite direction. It appears that bisimilarity and logical equivalence as stated in this paper are not the same,
as opposed to the classical variants defined by Simpson and Voorneveld \cite{modal-journal}.
This has to do with the limitations of explicit enumeration.

The variant formulation of relator in Definition
\ref{def:our-rel-lift} was chosen such that it yields program relations classically similar to the original ones,
yet allow us to constructively prove Propositions~\ref{prop:sim-trans} and~\ref{prop:relseq}.

\begin{proposition}\label{prop:sim-trans}
	If $P$ is similar to $Q$, and $Q$ is similar to $R$, then $P$ is similar to $R$.
\end{proposition}
\begin{proof}
	The proof follows ideas from \cite{opsem}, showing first that the union of two simulations is again a simulation.
	We then show that the reflexive-transitive closure $\overline{\mathcal{R}}^*$ of a simulation is a simulation too.
	Then, if $\overline{\mathcal{R}} \; P\;Q$ and $\overline{\mathcal{S}} \; Q\;T$, then $(\overline{\mathcal{R}} \cup \overline{\mathcal{S}})^* \; P\;T$, hence they are similar.
\end{proof}

For particular programming languages,
such as Plotkin's PCF \cite{PCF}, it can be classically proven that the resulting notion of $\Gamma$-similarity and $\Gamma$-bisimilarity is a \emph{precongruence} and \emph{congruence} respectively.
These kinds of proofs typically follow \emph{Howe's method} \cite{Howe96}. 
This is a desirable property, since any two programs related by a congruence are \emph{contextually equivalent}.
Moreover, in \cite{Matache19,MatacheS19} it is shown that for a continuation-passing style language, such a program equivalence coincides with contextual equivalence.

We can use the predicate liftings and constructive logic to study programming languages embedded into Agda.
Whether or not the resulting program equivalence is a congruence %(and hence implies contextual equivalence), 
is dependent on the operations of the programming language.
We try to approach such congruence results by studying common operations, such as sequencing and function application, and prove preservation over such operations in general. 
We leave to future work the question of whether the traditional methods such as Howe's method can be formalized and extended to work for embeddings of languages such as PCF into our generic programming language of denotations.

Program equivalences like applicative bisimilarity and logical equivalences are useful in regards to both applicability and formalization.
They reduce the burden of proof for showing equivalence of programs to either finding a bisimulation that relates the programs, or testing for behavioural properties at the appropriate type.

\section{Conclusion and Future Work}

In our development we have been working with a generic notion of
program given by elements of type $\Tm\;b\;\sigma$, for some label $b
: \{\val,\cpt \}$ and syntactic type $\sigma : \Ty$. These kind of
objects typically arise from the interpretation of effectful recursive
programs in denotational domain-theoretic models, or, alternatively,
from program evaluation w.r.t. some particular evaluation strategy. Our
Agda formalization includes the implementation of two different object
programming languages: fine-grained call-by-value PCF \cite{PCF,fine-grained} and an automaton-like state machine. Programs in
these languages can potentially be modelled as terms of our generic
language, and their behaviour analysed using the techniques developed
in this paper. 
%Our Agda code includes the interpretation of automata
%as coinductive trees, but the corresponding construction for PCF
%programs is left to future work. 
Our Agda code includes the interpretation of automata and PCF programs 
as coinductive trees. The full embedding of such computational constructs 
into our generic programming language of denotations is left for future work.
For the latter, we should take
inspiration from Benton et al.'s formalization of domain theory in Coq
\cite{BentonKV09} and Paviotti et al.'s full abstraction result for PCF in
guarded type theory \cite{PaviottiMB15}.

The type of coinductive trees we employ naturally captures an
``intensional'' view of possibly non-terminating computations, in the
sense that the user of a program may be allowed to observe the number
of computation steps. If this is the case, programs returning the same
values but with different computation time are considered different.
One might wish to change the encoding of programs and move to an
``extensional'' view of termination, so that intensional aspects
of computations such as computation time become unobservable.
We foresee a possible modification of the type of
trees involving the presence of the partiality
monad~\cite{ADK:parrpm,CUV:quodmw}. This modification in turn requires
the switch to proof assistants with support for higher %inductive
inductive types such as Cubical Agda~\cite{VMA:cubadt}. 

Many other effects besides the given examples can be described using this formalization too.
One such example is that of probability. It is possible to implement binary probabilistic choice, and define observations with the rational interval $O = [0,1)$ following \cite{modal-journal}. The inductive predicate lifting generated by an observation $p \in O$ checks whether the probability of satisfying a predicate is higher than $p$. The coinductive counterpart instead checks whether satisfaction of the lifted property or divergence together has at least a probability of $(1-p)$.
This example has been implemented, but has not been formally verified to be decomposable, due to the complexity of the necessary mathematical tools.

%\nocite{*}
\bibliographystyle{eptcs}
\bibliography{paper}
\end{document}